\documentclass[reqno,11pt]{article}
\usepackage{a4wide}
\usepackage{color,eucal,enumerate,mathrsfs,amsthm}
\usepackage[normalem]{ulem}
\usepackage{amsmath,amssymb,amsthm,epsfig,bbm,makecell}
\numberwithin{equation}{section}

\usepackage{amsfonts}
\usepackage{dsfont}
\usepackage{mathtools}
\usepackage{leftidx}
\usepackage{graphicx}
\usepackage{esint}
\usepackage{authblk}
\usepackage{multirow}

\usepackage[pdfborder={0 0 0}]{hyperref}  

\usepackage[latin1]{inputenc}
\usepackage[english]{babel}

\newtheorem{Definition}{Definition}[section]

\newtheorem{Theorem}[Definition]{Theorem}

\theoremstyle{definition}

\newtheorem{Remark}[Definition]{Remark}

\allowdisplaybreaks

\renewcommand{\H}{\mathbb{H}}

\newcommand{\N}{\mathbb{N}}

\newcommand{\R}{\mathbb{R}}

\newcommand{\mm}{{\mbox{\boldmath$m$}}}

\newcommand{\sfL}{{\sf L}}
\newcommand{\sfP}{{\sf P}}

\newcommand{\Kliminf}{K\kern-3pt-\kern-2pt\mathop{\rm lim\,inf}\limits}  
\renewcommand{\d}{{\mathrm d}}
\newcommand{\dt}{{\d t}}
\newcommand{\ddt}{{\frac \d\dt}}

\newcommand{\restr}[1]{\lower3pt\hbox{$|_{#1}$}} 

\newcommand{\nchi}{{\raise.3ex\hbox{$\chi$}}}

\newcommand{\fr}{\penalty-20\null\hfill$\blacksquare$}                      

\newcommand{\prob}[1]{\mathscr P(#1)}                   
\newcommand{\probt}[1]{\mathscr P_2(#1)}                   

\renewcommand{\mm}{\mathfrak m}                                

\newcommand{\HS}{{\lower.3ex\hbox{\scriptsize{\sf HS}}}}
\renewcommand{\H}[1]{{\rm Hess}(#1)}

\newcommand{\Ric}{{\rm Ric}}

\newcommand{\cons}{\mathscr{E}}

\title{A generalization of Costa's Entropy Power Inequality}
\author{Luca Tamanini \thanks{CEREMADE (UMR CNRS 7534), Universit\'e Paris Dauphine PSL, Place du Mar\'echal de Lattre de Tassigny, 75775 Paris Cedex 16, France and INRIA-Paris, MOKAPLAN, 2 Rue Simone Iff, 75012, Paris, France. email: tamanini@ceremade.dauphine.fr}}

\begin{document}

\maketitle

\begin{abstract}
Aim of this short note is to study Shannon's entropy power along entropic interpolations, thus generalizing Costa's concavity theorem. We shall provide two proofs of independent interest: the former by $\Gamma$-calculus, hence applicable to more abstract frameworks; the latter with an explicit remainder term, reminiscent of \cite{Villani06}, allowing us to characterize the case of equality.
\end{abstract}


\tableofcontents

\section{Introduction and statement of the result}

Given a random variable $X$ with density $u$ in $\R^n$, Shannon's entropy and entropy power are respectively defined as
\begin{equation}\label{eq:two def}
\mathcal H(u) := -\int_{\R^n} u\log u\,\d\mathcal{L}^n, \qquad \mathcal N(u) := \exp\Big(\frac{2}{n}\mathcal H(u)\Big),
\end{equation}
where $\mathcal{L}^n$ denotes the $n$-dimensional Lebesgue measure. In \cite{Costa85} Costa proved that Shannon's entropy power is concave along the heat flow, namely if $u$ is a non-negative probability density on $\R^n$ and $\sfP_t$ is the semigroup associated to the heat equation, i.e.\ 
\[
\frac{\partial}{\partial t}\sfP_t u = \Delta \sfP_t u,
\]
then
\begin{equation}\label{eq:costa}
\frac{\d^2}{\dt^2} \mathcal N(\sfP_t u) \leq 0, \qquad \forall t > 0.
\end{equation}
In fact, inequality is strict for all $t>0$ unless $u$ is an isotropic Gaussian distribution (in which case, equality holds for all $t>0$). The proof, originally quite involved, was eventually simplified in \cite{Dembo89,DCT91,Villani06}. This result plays an important role in information theory, as it allows for instance to deduce the Entropy Power Inequality (see \cite{Rioul10} for an exhaustive list of references), and it is also useful in connection to some functional inequalities, e.g.\ the dimensional logarithmic Sobolev inequality, as pointed out in \cite[Chapter 10]{ABCFGMRS00}.

The proof of \eqref{eq:costa} relies on De Bruijn's identity $\frac{\d}{\d t}\mathcal{H}(\sfP_t u) = \mathcal{I}(\sfP_t u)$ relating Shannon's entropy and Fisher information, the latter being defined as
\[
\mathcal I(u) := \int_{\R^n} \frac{|\nabla u|^2}{u}\,\d\mathcal{L}^n = \int_{\R^n} |\nabla\log u|^2 u\,\d\mathcal{L}^n.
\]
An interpolation problem strictly related with the heat semigroup is the so-called Schr\"odinger system, which reads as follows: given two probabilty measures $\mu = u\mathcal{L}^n$, $\nu = v\mathcal{L}^n$ and a parameter $T>0$, find two non-negative Borel functions $f^T$, $g^T$ (also called ``decomposition'') such that
\begin{equation}\label{eq:schr-system}
u = f^T\, \sfP_T g^T, \qquad v = g^T\, \sfP_T f^T.
\end{equation}
If this system is solvable, then $\rho_t^T := \sfP_t f^T \sfP_{T-t}g^T$ is a probability density which interpolates between $u$ at time $t=0$ and $v$ at time $t=T$ and we will eventually refer either to it or to $\mu_t^T := \rho_t^T\mathcal{L}^n$ as ``$T$-entropic interpolation''. It is worth mentioning that the heat flow is a particular entropic interpolation: indeed, if $\mu = u\mathcal{L}^n$ and $\nu = \sfP_T u\mathcal{L}^n$ (in the sequel, with a slight abuse of notation we will write $\nu = \sfP_T\mu$ for sake of brevity), then \eqref{eq:schr-system} is trivially solved by $f^T = u$ and $g^T = 1$.

From a physical point of view (see \cite{Leonard12} for a detailed discussion and \cite{Conforti17} for a more recent insight), $\mu$ and $\nu$ can be thought of as probability distributions of a cloud of independent Brownian particles observed at two different times and the entropic interpolation is the most-likely evolution between them. Up to reparametrization, $T$ can be interpreted either as a diffusion parameter or (as in the present paper) as the time interval between the two observations. In the former case, the link between Schr\"odinger problem and optimal transport appears naturally by large deviations theory, whereas in the latter the convergence of the $T$-entropic interpolation towards the heat flow as $T \to \infty$ is rather easy to guess. With this physical interpretation in mind, we recognize the following quantity
\[
\cons_T(\mu,\nu) := \frac{1}{2}\int_{\R^n}|v_t^T|^2\,\d\mu_t^T - \frac{1}{2}\mathcal I(\mu_t^T)
\]
as the total energy of the system and by the conservation of energy principle it is not surprising that $\cons_T(\mu,\nu)$ does not depend on $t$, although the right-hand side might a priori do (cf.\ \cite[Lemma 3.2]{GLRT19} for a rigorous proof). In the definition above
\[
v_t^T := \nabla\log \sfP_{T-t}g^T - \nabla\log \sfP_t f^T
\]
is the velocity field driving the $T$-entropic interpolation, since $\rho_t^T$ and $v_t^T$ are linked together by the continuity equation, namely $\partial_t\rho_t^T + {\rm div}(v_t^T\rho_t^T) = 0$, as proved for instance in \cite[Proposition 4.3]{GigTam18} in a very general framework.

\bigskip

In this paper we show that if we look at the entropy power $\mathcal N$ defined in \eqref{eq:two def} along the entropic interpolation $(\rho_t^T)_{t \in [0,T]}$ rather than along the heat flow on $[0,T]$, then a generalization of Costa's EPI \eqref{eq:costa} involving $\mathcal N$, $\mathcal I$ and $\cons_T$ can be deduced, at least if $(\rho_t^T)_{t \in [0,T]}$ interpolates between two suitable measures. For a more precise statement, let us first define $\mathcal{S}'$ as the space of $L^\infty$ functions with bounded support and the Schwartz-like space
\[
\mathcal{S}'' := \bigg\{ f \in L^\infty \cap C^\infty(\R^n) \,:\, \parbox{20em}{$\|x^\alpha D^\beta f\|_\infty < \infty$, $\forall \alpha \in \N^n$, $\beta \in \N^n \setminus \{0_n\}$, $|\log f(x)| \leq C(1+|x|^2)$ for some $C>0$} \bigg\},
\]
where $0_n$ is the null $n$-tuple; relying on these two spaces, let us introduce the class $\Upsilon_T$ of \emph{regular} constraints for \eqref{eq:schr-system} as follows
\[
\Upsilon_T := \Big\{ (\mu,\nu) \,:\, \mu,\nu \in \prob{\R^n} \textrm{ and } \exists f^T,\,g^T \in \mathcal{S'} \cup \mathcal{S''} \textrm{ solving \eqref{eq:schr-system}} \Big\}.
\]
By \cite[Proposition 2.1]{GigTam18} we know that all couples $(\mu,\nu)$ of absolutely continuous measures with bounded densities and supports belong to $\Upsilon_T$. Furthermore, also $(u\mathcal{L}^n,\sfP_T u\mathcal{L}^n) \in \Upsilon_T$ for any $u \in \mathcal{S}' \cup \mathcal{S}''$, since in the associated decomposition $g^T \equiv 1$ and constant functions belong to $\mathcal{S}''$: this is the reason behind the choice of excluding $\beta = 0_n$ (but still asking for $f \in L^\infty$) in the definition of $\mathcal{S}''$. As a consequence of this simple fact, any statement valid for all $(\mu,\nu) \in \Upsilon_T$ implies as a byproduct a particular statement for the heat flow: this is the case of Theorem \ref{thm:main-eucl} below, where Costa's EPI \eqref{eq:costa} appears as a particular case of its ``entropic'' version \eqref{eq:ent-costa}.

After this preamble, we can finally state our main results.



\begin{Theorem}\label{thm:main-eucl}
Let $(\mu,\nu) \in \Upsilon_T$ for some $T>0$ and denote by $(\rho_t^T)_{t \in [0,T]}$ the $T$-entropic interpolation between $\mu$ and $\nu$. Then $t \mapsto \mathcal N(\rho_t^T)$ belongs to $C([0,T]) \cap C^2((0,T))$ and it holds
\begin{equation}\label{eq:ent-costa}
\frac{\d^2}{\dt^2} \mathcal N(\rho_t^T) \leq \frac{4}{n^2}\mathcal N(\rho_t^T)\mathcal I(\rho_t^T) \cons_T(\mu,\nu).
\end{equation}
If $\nu = \sfP_T\mu$, then \eqref{eq:ent-costa} reduces to Costa's inequality \eqref{eq:costa}.
\end{Theorem}

As a consequence of the proof provided in Section \ref{sec:proof2}, we can also characterize the equality case in \eqref{eq:ent-costa}. Roughly speaking, it always reduces to the equality case in \eqref{eq:costa}, so that in particular inequality in \eqref{eq:ent-costa} is always strict along non-trivial entropic interpolations.

\begin{Theorem}\label{thm:equality}
With the same assumptions and notations as in Theorem \ref{thm:main-eucl}, there exists $t \in (0,T)$ where \eqref{eq:ent-costa} holds with equality if and only if either $\mu = u\mathcal{L}^n$, $\nu = \sfP_T u \mathcal{L}^n$ or $\nu = u\mathcal{L}^n$, $\mu = \sfP_T u \mathcal{L}^n$ for some isotropic Gaussian distribution $u$.
\end{Theorem}

As pointed out by Villani \cite{Villani06}, the concavity of Shannon's entropy power along the heat flow can also be deduced by relying on the so-called $\Gamma$-calculus, introduced by Bakry and \'Emery \cite{BakryEmery85} in the study of hypercontractive diffusions, although this approach does not allow to obtain a precise error term. This means that Costa's result holds not only in the Euclidean setting but also on Riemannian manifolds with non-negative Ricci curvature and even more generally (and with suitable modifications) on Riemannian manifolds with Ricci curvature bounded from below by some constant $K \in \R$. In this case \eqref{eq:costa} becomes
\begin{equation}\label{eq:costa-mfd}
\frac{\d^2}{\dt^2} \mathcal N(\sfP_t u) \leq -\frac{4K}{n}\mathcal N(\sfP_t u) \mathcal I(\sfP_t u), \qquad \forall t > 0,
\end{equation}
as recently proved in \cite{LiLi20}. In a completely analogous fashion, if we move from the Euclidean to the Riemannian framework, Theorem \ref{thm:main-eucl} reads as follows.

\begin{Theorem}\label{thm:main-mfd}
Let $(M,g)$ be an $m$-dimensional smooth, connected and complete Riemannian manifold without boundary, $V \in C^2(M)$ and $\mm = e^{-V}{\rm vol}$, where ${\rm vol}$ is the volume measure. Assume that for some $K \in \R$ and $n \geq m$ the Bakry-\'Emery Ricci tensor $\Ric_{V,n}$ satisfies the lower bound
\begin{equation}\label{eq:cd}
\Ric_{V,n} := \Ric_g + \H{V} - \frac{\nabla V \otimes \nabla V}{n-m} \geq Kg.
\end{equation}
Let $\mu,\nu \ll \mm$ be probability measures with bounded densities and supports and denote by $(\rho_t^T)_{t \in [0,T]}$ the $T$-entropic interpolation between them. Then $t \mapsto \mathcal N(\rho_t^T)$ belongs to $C([0,T]) \cap C^2((0,T))$ and it holds
\begin{equation}\label{eq:ent-costa-mfd}
\frac{\d^2}{\dt^2} \mathcal N(\rho_t^T) \leq \frac{4}{n^2}\mathcal N(\rho_t^T) \mathcal I(\rho_t^T) \cons_T(\mu,\nu) - \frac{2K}{n}\mathcal N(\rho_t^T) \Big(\int_M |v_t^T|^2\rho_t^T\,\d\mm + \mathcal I(\rho_t^T)\Big).
\end{equation}
If $M$ is compact and $\nu = \sfP_T\mu$, then \eqref{eq:ent-costa-mfd} reduces to \eqref{eq:costa-mfd}.
\end{Theorem}

Of course, this change of framework needs some remarks. Both in \eqref{eq:costa-mfd} and \eqref{eq:ent-costa-mfd} it is understood that the reference measure in the definition of $\mathcal H$, $\mathcal N$, and $\mathcal I$ is no longer $\mathcal{L}^n$ but $\mm$. As concerns the semigroup $\sfP_t$, it denotes the diffusion semigroup associated with the Witten Laplacian $\sfL = \Delta_g - \nabla V \cdot \nabla$, where $\Delta_g$ is the Laplace-Beltrami operator: this means that $\sfP_t u$ solves $\partial_t\sfP_t u = \sfL\sfP_t u$. It is such a semigroup that has to be considered in \eqref{eq:schr-system}, in the definition of the entropic interpolation $(\rho_t^T)$ as well as in \eqref{eq:costa-mfd}.

As regards the link between \eqref{eq:costa-mfd} and \eqref{eq:ent-costa-mfd}, it is still formally true that the former is a particular case of the latter, as we shall discuss in Section \ref{sec:heuristics}, but a rigorous proof is technical without compactness assumption. Already \eqref{eq:costa-mfd} requires more effort than \eqref{eq:costa}. The reason preventing us from saying that, in full generality, \eqref{eq:ent-costa-mfd} reduces to \eqref{eq:costa-mfd} when $\nu = \sfP_T\mu$ is the fact that \eqref{eq:ent-costa-mfd} will be proven under a boundedness assumption on the supports of $\mu,\nu$, whereas the support of $\sfP_T\mu$ is the whole manifold: thus $\nu = \sfP_T\mu$ is never satisfied, unless $M$ is compact.




\medskip

In the rest of the paper we shall give two different proofs of Theorem \ref{thm:main-eucl}. In Section \ref{sec:proof1} we present a first abstract argument based on $\Gamma$-calculus, which proves Theorems \ref{thm:main-eucl} and \ref{thm:main-mfd} at the same time. In Section \ref{sec:proof2} we provide a second (algebraic) proof of Theorem \ref{thm:main-eucl}, whence Theorem \ref{thm:equality} immediately follows.

\bigskip

\noindent{\bf Acknowledgements.} The author would like to thank G.\ Conforti for valuable suggestions.

\section{Proof by \texorpdfstring{$\Gamma$}{G}-calculus}\label{sec:proof1}

Otto and $\Gamma$-calculus are powerful tools: the former allows to obtain in a rather easy way heuristic explanations for technical statements on the Wasserstein space; the latter is an abstract formalism based on semigroup theory, which fits well to diffusions in both the Euclidean and Riemannian setting. For this reason in Section \ref{sec:heuristics} we first provide a heuristics for Theorems \ref{thm:main-eucl} and \ref{thm:main-mfd} to hold, while Section \ref{sec:rigorous} is devoted to the real proof by $\Gamma$-calculus.

\subsection{Heuristics}\label{sec:heuristics}

After Otto's seminal work \cite{Otto01}, it is well established that a formal Riemannian structure is associated with the Wasserstein space $(\probt{M},W_2)$. This means that we can treat $\mathcal H$ as a smooth function defined on a manifold (or, in an even simpler way, on $\R^n$) and $(\mu_t^T)_{t \in [0,T]}$ as a smooth trajectory on it. Furthermore, it is also well known that several PDEs on $M$ can be lifted to gradient flow equations on $\probt{M}$ w.r.t.\ the Wasserstein metric of suitable functionals: this is the case of the heat flow, which reads as the gradient flow of $-\mathcal{H}$, namely $\dot{\mu}_t = \nabla\mathcal{H}(\mu_t)$ where $\mu_t := \sfP_t\mu$, $\mu \in \probt{M}$. Since in $(\probt{M},W_2)$ it is more common to work with measures rather than the corresponding densities w.r.t.\ $\mm$, with a slight abuse we keep the same notations introduced before, e.g.\ $\mathcal H(\mu)$ denotes $\mathcal H(u)$ provided $\mu = u\mm$; analogously for $\mathcal{N}$ and $\mathcal{I}$.

\bigskip

As \eqref{eq:ent-costa-mfd} is a statement on the second derivative of $t \mapsto \mathcal N(t) := \mathcal N(\mu_t^T)$, let us differentiate it twice. The first derivative reads as
\[
\mathcal N'(t) = \mathcal N(t)\cdot \frac{2}{n}\langle\nabla \mathcal H(\mu_t^T),\dot{\mu}_t^T \rangle,
\]
so that the second derivative is given by
\[
\mathcal N''(t) = \mathcal N(t)\Big(\frac{4}{n^2}\big(\langle\nabla \mathcal H(\mu_t^T),\dot{\mu}_t^T \rangle\big)^2 + \frac{2}{n}\H{\mathcal H}(\dot{\mu}_t^T,\dot{\mu}_t^T) + \frac{2}{n}\langle\nabla \mathcal H(\mu_t^T),\ddot{\mu}_t^T\rangle\Big).
\]
A look at the dynamical aspects of entropic interpolations is required in order to move forward; more precisely, the ``acceleration'' of $t \mapsto \mu_t^T$ (or, in more geometric terms, the covariant derivative of $t \mapsto \dot \mu_t^T$ along $t \mapsto \mu_t^T$) has to be determined. The desired information is provided by the following Newton's law (see \cite{Conforti17})
\begin{equation}\label{eq:newton}
\ddot{\mu}_t^T = \frac{1}{2}\nabla|\nabla \mathcal H(\mu_t^T)|^2,
\end{equation}
so that the previous identity becomes
\[
\mathcal N''(t) = \mathcal N(t)\Big(\frac{4}{n^2}\big(\langle\nabla \mathcal H(\mu_t^T),\dot{\mu}_t^T \rangle\big)^2 + \frac{2}{n}\H{\mathcal H}(\dot{\mu}_t^T,\dot{\mu}_t^T) + \frac{2}{n}\H{\mathcal H}(\nabla \mathcal H(\mu_t^T),\nabla \mathcal H(\mu_t^T))\Big).
\]
Now we rely on the geometric structure of $(\probt{M},W_2)$ and, more specifically, on the role played by the curvature-dimension condition \eqref{eq:cd} in connection with the Boltzmann entropy $-\mathcal H$. After the seminal works \cite{Sturm06I, Sturm06II, Erbar-Kuwada-Sturm13} it is well known that \eqref{eq:cd} is equivalent to the $(K,n)$-convexity of $-\mathcal H$. Let us recall that a functional $\mathcal F : \probt{M} \to \R$ is said to be $(K,n)$-convex provided
\begin{equation}\label{eq:Km convex}
\H{\mathcal F} \geq K{\rm Id} + \frac{1}{n}\nabla \mathcal F \otimes \nabla \mathcal F.
\end{equation}
This information turns into an upper bound on $\H{\mathcal H}$, whence
\[
\mathcal N''(t) \leq \mathcal N(t)\Big(\frac{2}{n^2}\big(\langle\nabla \mathcal H(\mu_t^T),\dot{\mu}_t^T \rangle\big)^2 - \frac{2}{n^2}|\nabla \mathcal H(\mu_t^T)|^4 - \frac{2K}{n}\big(|\dot{\mu}_t^T|^2 + |\nabla \mathcal H(\mu_t^T)|^2\big)\Big).
\]
It we further note that $|\langle\nabla \mathcal H(\mu_t^T),\dot{\mu}_t^T \rangle| \leq |\nabla \mathcal H(\mu_t^T)||\dot{\mu}_t^T|$ by Cauchy-Schwarz inequality and use De Bruijn's identity, which reads as $|\nabla \mathcal H|^2 = \mathcal I$ in Otto's formalism (namely the Fisher information is nothing but the squared norm of the gradient of Shannon's entropy), we obtain
\[
\mathcal N''(t) \leq \mathcal N(t)\Big(\frac{2}{n^2}\mathcal{I}(\mu_t^T)\big(|\dot{\mu}_t^T|^2 - \mathcal{I}(\mu_t^T)\big)^2 - \frac{2K}{n}\big(|\dot{\mu}_t^T|^2 + \mathcal I(\mu_t^T) \big)\Big).
\]
It only remains to remark that the speed of a curve $(\mu_t)$ solving the continuity equation with drift $v_t$ can be expressed as $|\dot{\mu}_t| = \|v_t\|_{L^2(\mu_t)}$, so that we recognize $\int_M |v_t^T|^2\,\d\mu_t^T$ in \eqref{eq:ent-costa-mfd} as $|\dot{\mu}_t^T|^2$ and, as a consequence, $|\dot{\mu}_t^T|^2 - \mathcal{I}(\mu_t^T)$ as $2\cons_T(\mu,\nu)$.

The fact that \eqref{eq:costa-mfd} is a particular case of \eqref{eq:ent-costa-mfd} is not surprising, since, as already mentioned, the heat flow is a particular case of entropic interpolation. If $\mu = u\mm$ and $\nu = \sfP_Tu \mm$, then the Schr\"odinger system \eqref{eq:schr-system} is solved by $f^T = u$ and $g^T = 1$, whence
\[
|\dot{\mu}_t^T|^2 = \mathcal I(\mu_t^T) = \mathcal I(\sfP_t u) \qquad \textrm{and} \qquad \cons_T(\mu,\nu) = 0.
\]
Plugging these identities into \eqref{eq:ent-costa-mfd} yields \eqref{eq:costa-mfd}.

\begin{Remark}
{\rm
The heuristics described in this section actually applies to a wider class of variational problems, known as \emph{generalized Schr\"odinger problems} and introduced in \cite{GLR18}. Given $\mathcal F : \probt{M} \to \R \cup \{+\infty\}$ and $\mu,\nu \in \probt{M}$, they read as the following action minimizing problem
\[
\inf_{(\nu_t)_{t \in [0,T]}}\int_0^T \Big(\frac{1}{2}|\dot{\nu}_t|^2 + \frac{1}{2}|\nabla\mathcal F|^2(\nu_t)\Big)\dt,
\]
where the infimum runs over all paths $(\nu_t)_{t \in [0,T]}$ joining $\mu$ to $\nu$, and they indeed generalize the dynamic formulation of the entropic cost \`a la Benamou-Brenier (see \cite{GigTam19}). However, at present a rigorous investigation is not possible for these problems, except for the Schr\"odinger problem.
}\fr
\end{Remark}

\subsection{Proof of the result}\label{sec:rigorous}


Let us first discuss Theorem \ref{thm:main-eucl}. In order to turn the heuristic approach presented above into a precise one, rigorous counterparts of \eqref{eq:newton} and \eqref{eq:Km convex} are required. As concerns the former, we shall rely on the following formulas for the first and second derivatives of the entropy along entropic interpolations, computed for the first time in \cite{Leonard13}: 
\begin{subequations}
\begin{align}
\label{eq:firstder}
\frac{\d}{\d t}\mathcal H(\rho_t^T) & = -\int_{\R^n} \langle\nabla \rho_t^T,\nabla\vartheta_t^T \rangle\,\d\mathcal{L}^n,\\
\label{eq:secondder}
\frac{\d^2}{\d t^2}\mathcal H(\rho_t^T) & = -\int_{\R^n} \Big(\Gamma_2(\vartheta_t^T) + \Gamma_2(\log\rho_t^T)\Big)\rho_t^T\,\d\mathcal{L}^n,
\end{align}
\end{subequations}
where $\vartheta_t^T := \log\sfP_{T-t}g^T - \log\sfP_t f^T$ and $\Gamma_2$ is the iterated carr\'e du champ operator defined as
\[
\Gamma_2(\phi) := \frac{1}{2}\Delta|\nabla\phi|^2 - \langle\nabla\phi,\nabla \Delta\phi\rangle, \qquad \forall \phi \in C_c^\infty(\R^n).
\]
As $\sfP_t \phi$ is the convolution of $\phi$ with the $n$-dimensional Gaussian density having $0_n$ mean and $2t\mathrm{Id}_n$ as covariance matrix, $\sfP_t$ maps $\mathcal{S}' \cup \mathcal{S}''$ into $\mathcal{S}''$ for all $t>0$; actually, a stronger statement holds: for any $\phi \in \mathcal{S}' \cup \mathcal{S}''$, $\sfP_t\phi$ is, locally in $t \in (0,\infty)$, uniformly bounded by an integrable function and the same is true for $|\partial_t\sfP_t\phi|$ and $|\partial_t^2\sfP_t\phi|$. Therefore the same arguments that justify the (twice) differentiability of the entropy along the heat flow in Costa's EPI, namely of $t \mapsto -\int_{\R^n}\Phi(\sfP_t \phi)\,\d\mathcal{L}^n$ for $\phi \in \mathcal{S}''$ with $\Phi(z) := z\log z$, allow to deduce that 
\[
\alpha(s,t) := -\int_{\R^n}\Phi(\sfP_s f^T)\sfP_{T-t}g^T\,\d\mathcal{L}^n \qquad \textrm{and} \qquad\beta(s,t) := - \int_{\R^n}\Phi(\sfP_{T-s}g^T)\sfP_t f^T\,\d\mathcal{L}^n
\]
are $C^2$ on $(0,T) \times (0,T)$ and continuous up to the boundary; since $\mathcal{H}(\rho_t^T) = \alpha(t,t) + \beta(t,t)$, as a byproduct $t \mapsto \mathcal{H}(\rho_t^T)$ belongs to $C([0,T]) \cap C^2((0,T))$. The validity of \eqref{eq:firstder}, \eqref{eq:secondder} for all $t \in (0,T)$ is then a matter of computations (see the already cited \cite{Leonard13}).

On the other hand, it is not difficult to verify that in $\R^n$ it holds
\begin{equation}\label{eq:bochner-simple}
\Gamma_2(\phi) = |\H{\phi}|^2_\HS \geq \frac{1}{n}(\Delta\phi)^2, \qquad \forall \phi \in C^\infty_c(\R^n)
\end{equation}
and this replaces \eqref{eq:Km convex} with $K=0$. With this premise, the heuristic argument of the previous section becomes fully rigorous in the following way.

\begin{proof}[Proof of Theorem \ref{thm:main-eucl}]
As $t \mapsto \mathcal H(\rho_t^T)$ is $C([0,T]) \cap C^2((0,T))$, so is $t \mapsto \mathcal N(\rho_t^T)$. For sake of brevity set $\mathcal N(t) := \mathcal N(\rho_t^T)$ and write its first derivative as
\[
\mathcal N'(t) = \mathcal N(t) \cdot \frac{2}{n}\ddt \mathcal H(\rho_t^T),
\]
so that by \eqref{eq:firstder}, \eqref{eq:secondder} and \eqref{eq:bochner-simple} $\mathcal N''$ can be estimated as follows
\[
\begin{split}
\mathcal N''(t) & = \mathcal N(t)\bigg(\frac{4}{n^2}\Big(\ddt \mathcal H(\rho_t^T)\Big)^2 + \frac{2}{n}\frac{\d^2}{\dt^2}\mathcal H(\rho_t^T)\bigg) \\
& = \mathcal N(t)\bigg(\frac{4}{n^2}\Big(\int_{\R^n}\langle\nabla\rho_t^T,\nabla\vartheta_t^T\rangle\,\d\mathcal{L}^n\Big)^2 - \frac{2}{n}\int_{\R^n}\Big(\Gamma_2(\vartheta_t^T) + \Gamma_2(\log\rho_t^T)\Big)\rho_t^T \d\mathcal{L}^n\bigg) \\
& \leq \mathcal N(t)\bigg(\frac{4}{n^2}\Big(\int_{\R^n}\langle\nabla\rho_t^T,\nabla\vartheta_t^T\rangle\,\d\mathcal{L}^n\Big)^2 - \frac{2}{n^2}\int_{\R^n} \Big( (\Delta\vartheta_t^T)^2 + (\Delta\log\rho_t^T)^2\Big)\rho_t^T\d\mathcal{L}^n \bigg).
\end{split}
\]
Then by Jensen's inequality and integration by parts
\[
\begin{split}
\int_{\R^n} \Big( (\Delta\vartheta_t^T)^2 + (\Delta\log\rho_t^T)^2\Big)&\rho_t^T\d\mathcal{L}^n \geq \Big(\int_{\R^n} \Delta\vartheta_t^T\,\rho_t^T \d\mathcal{L}^n\Big)^2 + \Big(\int_{\R^n} \Delta\log\rho_t^T\,\rho_t^T \d\mathcal{L}^n\Big)^2 \\
& = \Big(-\int_{\R^n} \langle\nabla\rho_t^T,\nabla\vartheta_t^T\rangle\,\d\mathcal{L}^n\Big)^2 + \Big(-\int_{\R^n} |\nabla\log\rho_t^T|^2\,\rho_t^T \d\mathcal{L}^n\Big)^2 \\
& = \Big(\int_{\R^n} \langle\nabla\log\rho_t^T,\nabla\vartheta_t^T\rangle\rho_t^T\,\d\mathcal{L}^n\Big)^2 + \mathcal I(\rho_t^T)^2
\end{split}
\]
and plugging this inequality into the previous one yields
\[
\mathcal N''(t) \leq \mathcal N(t)\bigg(\frac{2}{n^2}\Big(\int_{\R^n}\langle\nabla\log\rho_t^T,\nabla\vartheta_t^T\rangle\rho_t^T\,\d\mathcal{L}^n\Big)^2 - \frac{2}{n^2}\mathcal I(\rho_t^T)^2 \bigg).
\]
By Cauchy-Schwarz inequality the first summand on the right-hand side can be controlled as
\begin{equation}\label{eq:cauchy-schwarz}
\Big(\int_{\R^n}\langle\nabla\log\rho_t^T,\nabla\vartheta_t^T\rangle\,\rho_t^T\d\mathcal{L}^n\Big)^2 \leq \int_{\R^n} |\nabla\log\rho_t^T|^2\rho_t^T\d\mathcal{L}^n \int_{\R^n} |\nabla\vartheta_t^T|^2\rho_t^T\d\mathcal{L}^n
\end{equation}
and this implies
\[
\mathcal N''(t) \leq \mathcal N(t) \cdot \frac{2}{n^2}\mathcal I(\rho_t^T)\bigg(\underbrace{\int_{\R^n} |\nabla\vartheta_t^T|^2 \rho_t^T\d\mathcal{L}^n - \mathcal I(\rho_t^T)}_{= 2\cons_T(\mu,\nu)}\bigg),
\]
whence \eqref{eq:ent-costa}. Finally, if $\nu = \sfP_T\mu$ (and say $\mu = u\mathcal{L}^n$), then as already said the associated Schr\"odinger system \eqref{eq:schr-system} is solved by $f^T = u$ and $g^T=1$. By the very definition of the total energy $\cons_T$, this implies $\cons_T(\mu,\nu) = 0$ and plugging this information into \eqref{eq:ent-costa} yields \eqref{eq:costa}.
\end{proof}

Let us now discuss Theorem \ref{thm:main-mfd}. The fact that $t \mapsto \mathcal H(\rho_t^T)$ belongs to $C([0,T]) \cap C^2((0,T))$ and \eqref{eq:firstder}, \eqref{eq:secondder} hold true for all $t \in (0,1)$ (with $\R^n$, $\mathcal{L}^n$ replaced by $M$, $\mm$ respectively) is justified by \cite[Proposition 4.8]{GigTam18}. On the other hand, the curvature-dimension assumption \eqref{eq:cd} is equivalent to the generalized Bochner inequality
\begin{equation}\label{eq:bochner}
\Gamma_2(\phi) \geq K|\nabla\phi|^2 + \frac{1}{n}(\sfL\phi)^2, \qquad \forall \phi \in C_c^\infty(M),
\end{equation}
which replaces \eqref{eq:Km convex}; of course, in the definition of $\Gamma_2$ on $M$, $\sfL$ substitutes $\Delta$. After this digression, the proof of Theorem \ref{thm:main-mfd} is a minor modification of the previous one.

\begin{proof}[Proof of Theorem \ref{thm:main-mfd}] 
Computing $\mathcal N''$ as in the proof of Theorem \ref{thm:main-eucl} and using \eqref{eq:bochner} instead of \eqref{eq:bochner-simple} we get
\[
\begin{split}
\mathcal N''(t) & \leq \mathcal N(t)\bigg(\frac{4}{n^2}\Big(\int_M \langle\nabla\rho_t^T,\nabla\vartheta_t^T\rangle\,\d\mm\Big)^2 - \frac{2}{n^2}\int_M \Big( (\sfL\vartheta_t^T)^2 + (\sfL\log\rho_t^T)^2\Big)\rho_t^T\d\mm \\
& \qquad\qquad\qquad - \frac{2K}{n}\int_M \Big( |\nabla\vartheta_t^T|^2 + |\nabla\log\rho_t^T|^2\Big)\rho_t^T\d\mm \bigg).
\end{split}
\]
Now it suffices to follow the same argument as in the proof of Theorem \ref{thm:main-eucl} (replacing $\Delta$, $\mathcal{L}^n$ with $\sfL$, $\mm$ respectively, the same integration by parts formula is valid) and keep track of the additional term until the end to obtain \eqref{eq:ent-costa-mfd}.

If $M$ is compact and $\nu = \sfP_T\mu$ with $\mu = u\mm$, then $\cons_T(\mu,\nu) = 0$ as in Theorem \ref{thm:main-eucl} and moreover $\rho_t^T = \sfP_t u$, $|\nabla\vartheta_t^T| = |\nabla\log\sfP_t u|$.
\end{proof}

\begin{Remark}
{\rm
A further way to see that \eqref{eq:ent-costa-mfd} implies \eqref{eq:costa-mfd} relies on the long-time behavior of $T$-entropic interpolations and the energy $\cons_T(\mu,\nu)$, investigated in \cite{ConTam19}. From an intuitive point of view, the more the time parameter $T$ grows, the less the final condition $v = g^T\sfP_T f^T$ in \eqref{eq:schr-system} is influent, so that in the limit \eqref{eq:schr-system} is in fact a decoupled system and the $T$-entropic interpolation is nothing but the heat flow starting at $\mu$.

Under the further assumption that $K \geq 0$ and $\mm$ is a probability (for $K>0$ this is always true thanks to \cite[Theorem 4.26]{Sturm06I}), by \cite[Theorem 1.2]{ConTam19} we know that $\cons_T(\mu,\nu) \to 0$ as $T \to \infty$. Moreover, if $\mu = u\mm$ and $\nu = v\mm$, by \cite[Lemma 3.6]{ConTam19} we also know that
\[
\lim_{T \to \infty}f^T = u, \qquad \lim_{T \to \infty}g^T = v, \qquad \lim_{T \to \infty} \sfP_T f^T = \lim_{T \to \infty} \sfP_T g^T = 1,
\]
where all limits are in $L^p(\mm)$ for any $p \in [1,\infty)$, so that $\rho_t^T \to \sfP_t u$ in $L^p(\mm)$ as $T \to \infty$. Therefore it is intuitively clear that \eqref{eq:costa-mfd} can be recovered as the long-time limit of \eqref{eq:ent-costa-mfd}, since as $T \to \infty$ we expect the first term on the right-hand side of \eqref{eq:ent-costa-mfd} to vanish, whereas
\[
\lim_{T \to \infty}\int_M |v_t^T|^2\rho_t^T\,\d\mm + \mathcal{I}(\rho_t^T) = 2\mathcal{I}(\sfP_t u).
\]
To turn this sketch of proof into a rigorous demonstration, one should only pass through an integrated version of \eqref{eq:ent-costa-mfd} and argue by dominated convergence.
}\fr
\end{Remark}


\section{Proof of Theorem \ref{thm:main-eucl} with deficit}\label{sec:proof2}

In this section we shall give a direct proof of Theorem \ref{thm:main-eucl} with an (almost) exact error term, in the same spirit of \cite{Villani06}. This means that we shall put aside \eqref{eq:bochner-simple} and argue by explicit computations: the disadvantage is the validity of the approach only in the Euclidean setting, but as an advantage we are able to characterize the case of equality in \eqref{eq:ent-costa}, thus proving Theorem \ref{thm:equality}.

Inspired by \cite{Villani06}, let us first observe that, thanks to the computations carried out in the proof of Theorem \ref{thm:main-eucl} (and in particular thanks to the formula for $\mathcal N''$), \eqref{eq:ent-costa} is equivalent to
\begin{equation}\label{eq:alternative-ent-costa}
\int \Big(\Gamma_2(\vartheta_t^T) + \Gamma_2(\log\rho_t^T)\Big)\rho_t^T\,\d\mathcal{L}^n - \frac{2}{n}\Big(\int\langle\nabla\rho_t^T,\nabla\vartheta_t^T\rangle\,\d\mathcal{L}^n\Big)^2 + \frac{2}{n}\mathcal I(\rho_t^T)\cons_T(\mu,\nu) \geq 0
\end{equation}
and define
\[
A_1(\lambda) := \sum_{i,j=1}^n\int\big(\partial_{ij}\vartheta_t^T + \lambda\delta_{ij}\big)^2\rho_t^T\,\d\mathcal{L}^n, \qquad A_2(\eta) := \sum_{i,j=1}^n\int \big(\partial_{ij}\log\rho_t^T + \eta\delta_{ij}\big)^2\rho_t^T\,\d\mathcal{L}^n.
\]
By expanding $A_1$ as a binomial in $\lambda$, using integration by parts and the fact that $\rho_t^T$ is a probability density, we get
\[
\begin{split}
A_1(\lambda) & = \sum_{i,j=1}^n\int_{\R^n}\big(\partial_{ij}\vartheta_t^T\big)^2\rho_t^T\,\d\mathcal{L}^n - 2\sum_{i=1}^n\lambda\int_{\R^n} \partial_i\rho_t^T \partial_i \vartheta_t^T \,\d\mathcal{L}^n + \lambda^2 n \\
& = \int_{\R^n} \Gamma_2(\vartheta_t^T)\rho_t^T\,\d\mathcal{L}^n - 2\lambda\int_{\R^n}\langle\nabla\rho_t^T,\nabla\vartheta_t^T\rangle\,\d\mathcal{L}^n + \lambda^2 n,
\end{split}
\]
the second identity being motivated by Bochner's identity in $\R^n$, namely
\[
\sum_{i,j=1}^n(\partial_{ij}\phi)^2 = |\H{\phi}|^2_\HS = \frac{1}{2}\Delta|\nabla\phi|^2 - \langle\nabla\phi,\nabla\Delta\phi\rangle = \Gamma_2(\phi).
\]
In particular
\[
\lambda^* = \frac{1}{n}\int_{\R^n}\langle\nabla\rho_t^T,\nabla\vartheta_t^T\rangle\,\d\mathcal{L}^n \quad \Rightarrow \quad A_1(\lambda^*) = \int_{\R^n} \Gamma_2(\vartheta_t^T)\rho_t^T\,\d\mathcal{L}^n - \frac{1}{n}\Big(\int_{\R^n}\langle\nabla\rho_t^T,\nabla\vartheta_t^T\rangle\,\d\mathcal{L}^n\Big)^2.
\]
Arguing in the same way for $A_2$, we obtain
\[
A_2(\eta) = \int_{\R^n} \Gamma_2(\log\rho_t^T)\rho_t^T\,\d\mathcal{L}^n - 2\eta \mathcal I(\rho_t^T) + \eta^2 n
\]
with
\[
\eta^* = \frac{1}{n}\mathcal I(\rho_t^T) \quad \Rightarrow \quad A_2(\eta^*) = \int_{\R^n} \Gamma_2(\log\rho_t^T)\rho_t^T\,\d\mathcal{L}^n - \frac{1}{n} \mathcal I(\rho_t^T)^2.
\]
Hence
\[
\begin{split}
A_1(\lambda^*) + A_2(\eta^*) & = \int_{\R^n} \Gamma_2(\vartheta_t^T)\rho_t^T\,\d\mathcal{L}^n + \int_{\R^n} \Gamma_2(\log\rho_t^T)\rho_t^T\,\d\mathcal{L}^n - \frac{2}{n}\Big(\int_{\R^n}\langle\nabla\rho_t^T,\nabla\vartheta_t^T\rangle\,\d\mathcal{L}^n\Big)^2 \\
& \qquad + \frac{1}{n}\Big(\int_{\R^n}\langle\nabla\rho_t^T,\nabla\vartheta_t^T\rangle\,\d\mathcal{L}^n\Big)^2 - \frac{1}{n} \mathcal I(\rho_t^T)^2
\end{split}
\]
and by \eqref{eq:cauchy-schwarz} and the very definition of $\cons_T(\mu,\nu)$
\begin{equation}\label{eq:rancid}
\frac{1}{n}\Big(\int_{\R^n}\langle\nabla\rho_t^T,\nabla\vartheta_t^T\rangle\,\d\mathcal{L}^n\Big)^2 - \frac{1}{n} \mathcal I(\rho_t^T)^2 \leq \frac{2}{n}\mathcal I(\rho_t^T)\cons_T(\mu,\nu).
\end{equation}
Plugging this inequality into the previous identity exactly yields \eqref{eq:alternative-ent-costa}, since trivially $A_1(\lambda^*) + A_2(\eta^*) \geq 0$. Note that if equality occurs in \eqref{eq:alternative-ent-costa} for some $t$, then in particular equality must hold in \eqref{eq:rancid} and this is true if and only if \eqref{eq:cauchy-schwarz} is actually an identity. As Cauchy-Schwarz inequality is an equality only for parallel vectors, this means that
\[
\textrm{either} \quad \nabla\log\rho_t^T = \nabla\vartheta_t^T, \qquad \textrm{or} \quad \nabla\log\rho_t^T = -\nabla\vartheta_t^T.
\]
By definition of $\rho_t^T$ and $\vartheta_t^T$, this means that
\[
\textrm{either} \quad \nabla\log\sfP_t f^T = 0, \qquad \textrm{or} \quad \nabla\log\sfP_{T-t} g^T = 0,
\]
namely either $f^T$ or $g^T$ is constant and this implies that the entropic interpolation $(\mu_t^T)_{t \in [0,T]}$ is in fact a forward or backward heat flow. Therefore the case of equality in the entropic EPI \eqref{eq:ent-costa} reduces to equality in Costa's EPI \eqref{eq:costa} and the latter is already well understood (see e.g.\ \cite[Theorem 3]{Costa85}).

\bibliographystyle{siam}
{\small
\bibliography{biblio}}

\end{document}